\documentclass[letterpaper, 10 pt, conference]{ieeeconf}  
                        
\IEEEoverridecommandlockouts                              
\usepackage[english]{babel}
\usepackage{flushend,booktabs}
\usepackage{bigints}
\usepackage{graphicx} 
\usepackage{psfrag}
\usepackage{amsfonts}
\usepackage{amssymb}
\usepackage{amsmath}
\usepackage{graphicx}
\usepackage{psfrag}
\usepackage{ae,epsfig,url}
\usepackage{verbatim} 
\usepackage{enumerate} 
\usepackage{epstopdf}
\usepackage{tikz}
\usetikzlibrary{shapes,arrows}
\usetikzlibrary{positioning,fit}
\usepackage{color}
\usetikzlibrary{positioning}
\usepackage{caption}
\usepackage{subcaption}
\usepackage{tikz}
\usetikzlibrary{shapes,arrows}
\usetikzlibrary{positioning,fit}
\usepackage{bm}

\usetikzlibrary{calc}
\usetikzlibrary{matrix}

\usepackage[font=small,labelfont=bf]{caption}

\newcounter{teocount}

\newcounter{remcount}

\newcounter{excount}

\newtheorem{remm}[remcount]{Remark}

\newtheorem{theorem}[teocount]{Theorem}

\newtheorem{ex}[excount]{Example}

\newenvironment{remark}{\begin{remm}\rm }{\hfill \hspace*{1pt} \hfill
$\star$\end{remm}}

\author{Ross Drummond and Yang Zheng
\thanks{This work was supported by the UK Engineering and Physical Sciences Research Council (EPSRC) through grant EP/P005411/1-``Structured electrodes for improved energy storage''.}
\thanks{R. Drummond is with the Department of Engineering Science, University of Oxford, 17 Parks Road, OX1 3PJ, Oxford, United Kingdom. Email: ross.drummond@ eng.ox.ac.uk.}
\thanks{Y. Zheng is with SEAS and CGBC at Harvard University, Cambridge, MA 02138. Email: zhengy@g.harvard.edu.
}
}

\title{\bf 
Impact of Disturbances on Mixed Traffic Control with Autonomous Vehicles
}

\begin{document}
\maketitle
\thispagestyle{plain}
\pagestyle{plain}

\begin{abstract}
This paper investigates the impact of disturbances on controlling an autonomous vehicle to smooth mixed traffic flow in a ring road setup.  
By exploiting the ring structure of this system, 
it is shown that velocity perturbations impacting any vehicle on the ring enter an uncontrollable and marginally stable mode defined by the sum of relative vehicle spacings. These disturbances are then integrated up by the system and cannot be unwound via controlling the autonomous vehicle. In particular, if the velocity disturbances are zero-mean Gaussians, then the traffic flow on the ring will undergo a random walk with the variance growing indefinitely and independently of the control policy applied. In contrast, the impact of acceleration disturbances is benign as these disturbances do no enter the uncontrollable mode, meaning that they can be easily regulated using the autonomous vehicle. Our results support and complement the existing theoretic analysis and field experiments.

\end{abstract}


\section{Introduction}

Fuelled by recent advances in data processing and decision making under uncertainty, it is  projected that fleets of autonomous vehicles will be appearing on public roads in the near future, using advanced sensors and huge databases of LiDAR, mapping and camera data to manoeuvre around the urban cityscape on their own accord \cite{litman}. Whilst many challenges still need to be overcome before this vision is realised (including regulatory, controller robustness, and ethical issues), the promise of a safer and more efficient transportation system offered by vehicle autonomy continues to drive this technology forward at a breakneck speed. 

One particularly promising avenue for autonomous vehicles is in improving the efficiency of traffic flows, with the autonomous vehicles acting as control actuators on the traffic and having the crucial benefit of being able to communicate with the vehicles around them~\cite{guanetti2018control,li2017dynamical}. In this context, the potential for autonomy 
to dampen  ``stop-and-go'' waves in the traffic flow is particularly exciting \cite{exp,laval}. These waves result in concertina like blockages in the traffic and form naturally in congested roads even when no obstacles are present~\cite{sugiyama2008traffic}. Unfortunately, 
the efficiency of dampening these waves using classical traffic control methods like variable speed limits might be limited~\cite{smulders}. 
Instead, by carefully exploiting the information of the traffic flow gathered by their sensors, autonomous vehicles are promising to  efficiently dampen them by predicting the response of the traffic network and reacting accordingly.

The experiment of \cite{exp} represents 
one striking demonstration of the potential of autonomous vehicles to damp ``stop-and-go'' traffic waves. In this work, 20+ human drivers 
were told to drive around in a ring at a constant speed, with the caveat being that one of the cars was capable of being driven autonomously. With the 
control of the autonomous vehicle 
turned off, ``stop-and-go'' waves quickly formed in the ring, causing 
the vehicles to periodically grind to a full stop. Upon activating the 
autonomous vehicle 
(which employed a simple ``slow-in fast-out'' control policy that slowed the ring speed down before accelerating to the desired equilibrium velocity),  these waves were quickly dissipated and the vehicle ring stabilised around the desired constant velocity. Thus, with a rather simple control policy applied through some agent (the autonomous vehicle), the experiment in \cite{exp} demonstrated the high potential for autonomous vehicles to smooth traffic flow, as predicted by theory~\cite{cui, yang}.

The powerful implications of the experiment in~\cite{exp} has motivated recent research, including \cite{cui, yang, giammarino2019weaker}, into analysing its results from a control theoretic perspective. Doing so provides insight into the formation of the ``stop-and-go'' waves and answers how there could exist a control law that could optimally dissipate them. The most prominent of these studies is \cite{yang} which developed an analytic controllability decomposition of the experimental set-up in \cite{exp} by exploiting the ring structure of the traffic flow. One main result of \cite{yang} showed that a linearised model of the experiment in \cite{exp} only had one uncontrollable mode (related to the length/centroid of the ring) when the autonomous vehicle was used as the controller, with this uncontrollable mode being marginally stable. Accordingly, as long as this mode remained unperturbed, a controller can always be found to stabilise the system. In this way, the authors in \cite{yang} exploited analytic results from linear control theory to explain the powerful experimental observations of \cite{exp}.

Building upon the results of \cite{yang}, this paper also considers the controllability of traffic flow rings but focuses on the potential pitfalls that could be encountered due to disturbances entering the system's uncontrollable, marginally stable mode. The main result is given in Theorem \ref{prop1} which states that, unless the overall disturbance is zero, then velocity disturbances impacting any vehicle in the ring will be integrated by the uncontrollable mode and, furthermore, there does not exist a controller capable of dissipating them. In contrast, the impact of acceleration disturbances on the ring stability are more benign and do not enter this 
uncontrollable mode. To indicate the significance of this result, it implies that if it is assumed the system is subject to zero-mean, Gaussian velocity disturbances, then the sum of the spacings of the vehicles on the ring will undergo a random walk, an undesirable response with the variance of the vehicles growing indefinitely no matter what control policy is applied. The purpose of this paper is to generalise the analysis of \cite{yang} by accounting for external disturbances and, since only acceleration disturbances are encountered in practice, to provide the theory supporting the field experiments of \cite{exp}.

The paper is structured as follows. Section \ref{sec:mod} details the mixed traffic flow model. Section \ref{sec:thm} states the main theorem of the paper on the impact of disturbances on the  stability of mixed traffic flow rings. Section \ref{sec:sim} contains numerical simulations of the traffic flow ring under various disturbances verifying the statement of the main theorem, and we conclude the paper in Section~\ref{sec:con}.

\section{Model for mixed traffic flow in a ring}\label{sec:mod}

To arrive at the structural controllability result of Theorem \ref{prop1}, the vehicle model is first introduced before the model for the vehicle ring. The model derivation of this section closely follows the procedure of \cite{yang}, and, in fact,  the results of this paper can be seen as an extension of \cite{yang} to account for the impact of disturbances.

\subsection{Linearised human-driver model}
The typical approach to model a human-driven vehicle is to apply Newton's laws to the vehicle and then express its acceleration $\dot{v}_i(t)$ as a function of its spacing $s_i$, relative velocity $\dot{s}_i(t)$ between its own and the preceding vehicle $i-1$ as well as its velocity $v_i(t)$
\begin{align}\label{v_nonlin}
\dot{v}_i(t) = F(s_i(t),\dot{s}_i(t), v_i(t)),
\end{align}
with $ F= \mathbb{R} \times \mathbb{R} \times \mathbb{R} \to \mathbb{R}$ known as the forcing function \cite{helbing,orosz}. This equation forms the basis of the considered model for human-driven vehicles, however, for the controllability analysis of Theorem \ref{prop1}, a linear model is required. Thus, we linearise \eqref{v_nonlin} around some equilibrium. Denote $s^*$ and $v^*$ as, respectively, the equilibrium spacing and velocities of all the vehicles satisfying $\tilde{s}_i(t) = s_i(t)-s^*$, $\tilde{v}_i(t) = v_i(t)-v^*$ and
\begin{align}
F(s^*,0,v^*) = 0.
\end{align}

The time derivative of the local perturbations around these equilibria relates to the linearised model dynamics of interest to this paper. For the localised spacing dynamics, these local dynamics are $\dot{\tilde{s}}_i(t) = \dot{s}_i(t) = v_{i-1}(t)-v_i(t) = \tilde{v}_{i-1}(t)-\tilde{v}_i(t)$ and those for the velocities are obtained by taking a Taylor expansion of the forcing term, so 
\begin{align}
\! \dot{\tilde{v}}_i(t) = \dot{v}_i(t) =  \underbrace{F(s^*,0,v^*)}_{ =\, 0} + \frac{\partial F}{\partial s}\tilde{s}_i 
 + \frac{\partial F}{\partial \dot{s}}\dot{\tilde{s}}_i 
 + \frac{\partial F}{\partial v}\tilde{v}_i.
\end{align}
Defining $\alpha_1 = \frac{\partial F}{\partial s}$,  $\alpha_2 =  \frac{\partial F}{\partial \dot{s}}-\frac{\partial F}{\partial v}$ and $\alpha_3 =  \frac{\partial F}{\partial \dot{s}}$, then the linearised human-driven vehicle can be expressed in the following state-space form
\begin{align}\label{sys}
\begin{cases} \dot{\tilde{s}}_i(t) = \tilde{v}_{i-1}(t)-\tilde{v}_i(t), \\
\dot{\tilde{v}}_i(t) = \alpha_1\tilde{s}_{i}(t)-\alpha_2\tilde{v}_i(t) + \alpha_3 \tilde{v}_{i-1}(t).
 \end{cases}
\end{align}

The next step is to choose the forcing function $F(s_i(t),\dot{s}_i(t), v_i(t))$. In this paper, the modelling framework of the optimal velocity model (OVM) is adopted \cite{orosz,bando}, with a forcing function  
\begin{align}
F(s_i(t), \dot{s}_i(t), v_i(t)) = \alpha(V(s_i(t))-v_i(t))+\beta \dot{s}_i(t).
\end{align}
Here, $\alpha>0$ magnifies the difference between the driver's spacing dependent desired velocity and that which they are driving whilst $\beta >0$ incorporates the velocity difference between the vehicle and the preceding vehicle into the driver's actions \cite{cui,wilson}. Typically, the desired spacing dependent velocity of the driver is modelled as the piece-wise function  \cite{wilson}
\begin{align}
V(s) = \begin{cases} 0 & s \leq s_{st}, \\ f_v(s), & s_{\text{st}}<s<s_{\text{go}}, \\
v_{\text{max}},  & s \geq s_{\text{go}},\end{cases}
\end{align}
which stops the vehicle when the spacing is small (to avert a crash), saturates the velocity at $v_{\text{max}}$ and follows a monotonically decreasing function as the spacing approaches the stopping value $s_{\text{st}}$.  A typical form of this function is 
\begin{align}
f_v(s)= \frac{v_{\text{max}}}{2}\left(1-\cos\left(\pi\left(\frac{s-s_{\text{st}}}{s_{\text{go}}-s_{\text{st}}}\right)\right)\right).
\end{align}
When using the OVM, the equilibrium velocity is given by $v^* = V(s^*)$ and the local perturbations around this equilibrium are governed by \eqref{sys} with coefficients 
$\alpha_1 = \alpha \frac{dV}{ds}\Big|_{s = s^*}$, $\alpha_2 = \alpha+\beta$ and $\alpha_3 = \beta$ where $\frac{dV}{ds}\big|_{s = s^*}$ is the slope of the spacing dependent driver velocity evaluated at the equilibrium spacing $s^*$. When the spacing is not too extreme, lying in the region $s_{\text{st}} < s^*< s_{\text{go}}$, then $\frac{dV}{ds}\big|_{s = s^*} = \frac{f_v(s)}{ds}\big|_{s = s^*}$ where
\begin{align}
\frac{f_v(s)}{ds}= \frac{v_{\text{max}}\pi}{2(s_{\text{go}}-s_{\text{st}})}\sin\left(\pi\left(\frac{s-s_{\text{st}}}{s_{\text{go}}-s_{\text{st}}}\right)\right).
\end{align}
This will be the region considered in the simulation of Section \ref{sec:sim}.

\begin{remark}
As emphasized in \cite{yang}, the structural nature of the following  controllability result in Theorem \ref{prop1} means that it can be easily generalised to other human-driven vehicle models, such as the intelligent driver model. Throughout this paper, due to some technicality, it is assumed that $\alpha_1-\alpha_2\alpha_3 + \alpha_3 \neq 0$ (see~\cite[Theorem 2]{yang} for more discussions).
\end{remark}


\subsection{Dynamics of a mixed traffic flow ring}
Consider a one lane traffic flow ring of $n$ vehicles with one of the vehicles in the ring being autonomous and the remaining $n-1$ being human driven with dynamics as in \eqref{sys}. In a state-space form, this system is described by
\begin{align}\label{sys1}
\dot{x}(t) = Ax(t) + Bu(t)+d(t), 
\end{align}
where
\begin{align}
d(t) = \begin{bmatrix} d^v_1(t), & d^a_1(t),  & \dots, & d^v_{n}(t), & d^a_{n}(t) \end{bmatrix}^T
\end{align}
is a time-varying vector of dimension $2n$ that contains the disturbances impacting the velocity $d^v \in \mathbb{R}^n$ and acceleration $d^a \in \mathbb{R}^n$ that could be due to external forces like wind gusts or enter simply from the stochastic nature of human drivers; $x \in \mathbb{R}^{2n}$ is the state of the traffic flow ring containing each vehicles relative spacings and velocities around the equilibrium; $u \in \mathbb{R}$ is the acceleration input of the single autonomous vehicle which acts as the control action. The considered system has a similar structure to~\cite{yang} except it also includes the disturbances $d(t)$, whose impact upon the stability of the traffic ring is the main 
interest of this paper.  

The state-space matrices of the system are of the form 
\begin{align}\label{state_mats}
A = \begin{bmatrix} C_1 & 0 & \dots & \dots & 0 & C_2 \\ 
A_2 & A_1 & 0  & \dots & \dots & 0 \\ 0 & A_2 & A_1 & 0 & \dots & 0 \\
\vdots & \ddots & \ddots & \ddots & \ddots & \vdots \\
0 & \dots & 0 & A_2 & A_1 & 0 \\ 
0 & \dots & \dots & 0 & A_2 & A_1 \end{bmatrix},B = \begin{bmatrix}
B_1 \\ B_2 \\ B_2 \\ \vdots \\ B_2
\end{bmatrix}\!,
\end{align}
and 
\begin{subequations}\begin{align}
A_1  &= \begin{bmatrix} 0 & -1 \\ \alpha_1 & -\alpha_2\end{bmatrix},
A_2 = \begin{bmatrix} 0 & 1 \\ 0 & \alpha_3\end{bmatrix}, \\
C_1  & = \begin{bmatrix} 0 & -1 \\ 0 & 0\end{bmatrix},
C_2 = \begin{bmatrix} 0 & 1 \\ 0 & 0\end{bmatrix}, \label{eq:C1C2}\\
B_1 &= \begin{bmatrix} 0\\ 1\end{bmatrix},
B_2 = \begin{bmatrix} 0 \\ 0\end{bmatrix},
\end{align}\end{subequations}
where the matrices $C_1$ and $C_2$ in~\eqref{eq:C1C2} correspond to the autonomous vehicle's dynamics. By focusing on the case of traffic flows on rings, as demonstrated in the experiment of \cite{exp}, a block circulant structure can be seen to be emerging in the state-space matrices \eqref{state_mats} which is exploited to obtain analytic controllability results. In this paper, we focus on the linearised model~\eqref{sys1} to analyze the impact of disturbances. The 
extensive simulations with nonlinear car-following models in~\cite{yang,wang2020controllability} suggest that the linearisation analysis captures the behavior of the nonlinear traffic system around the equilibrium state.

\section{Main results: Velocity perturbations enter an uncontrollable, marginally stable mode of the system.} \label{sec:thm}

The main result of this paper is the following theorem stating that for vehicles being driven at equilibrium in a ring, velocity disturbances to any vehicle will enter an uncontrollable~\cite{kalman}, marginally stable mode of the interconnected system (unless its net effect is zero) which just integrates up this disturbance and cannot be unwound by any feedback controller. A consequence of this result (as highlighted in Figure \ref{fig:fig_pert_s}) is that if it is assumed that the velocity disturbances are Gaussian, then the ring's total relative spacing will undergo a random walk which can not be controlled. As an implication, this implies that there exists a velocity disturbance which can destabilise a traffic flow ring controlled by an autonomous vehicle. In contrast, acceleration disturbances are benign and do not enter the uncontrollable mode, explaining the performance of the control used in \cite{exp} which would be subject to external disturbances, e.g., from wind gusts or random adjustments in human driving.

\begin{theorem}\label{prop1}
{Consider the mixed traffic system in a ring road
with one autonomous vehicle and $n-1$ human-driven vehicles given by~\eqref{sys1}.} Unless
\begin{align}\label{noise_cond}
\sum_{i = 1}^{n} d^v_{i} (t) = 0,~\forall t \in [0, \infty)
\end{align}
then any velocity disturbances $d^v$ will be integrated by an uncontrollable, marginally stable mode of system \eqref{sys1}. 
\end{theorem}
\begin{proof}
The proof builds upon the decomposition of \cite{yang} but focuses on the impact of external disturbances. 

Define a new virtual input $\hat{u}(t) = u(t)  -(\alpha_1\tilde{s}_1(t)-\alpha_2\tilde{v}_1(t)+\alpha_3 \tilde{v}_n(t))$ corresponding to the difference between the acceleration value when the vehicle is human-driven and the actual control value. With this new variable, \eqref{sys1} can be equivalently written as
\begin{align}\label{sys2}
\dot{x}(t) = \hat{A}x(t) + B\hat{u}(t)+d(t), 
\end{align}
where
\begin{align*}
A = \begin{bmatrix} A_1 & 0 & \dots & \dots & 0 & A_2 \\ 
A_2 & A_1 & 0  & \dots & \dots & 0 \\ 0 & A_2 & A_1 & 0 & \dots & 0 \\
\vdots & \ddots & \ddots & \ddots & \ddots & \vdots \\
0 & \dots & 0 & A_2 & A_1 & 0 \\ 
0 & \dots & \dots & 0 & A_2 & A_1 \end{bmatrix}, B = \begin{bmatrix}
B_1 \\ B_2 \\ B_2 \\ \vdots \\ B_2
\end{bmatrix}.
\end{align*}
As controllability is independent of state feedback for linear systems\footnote{This is a classical result; see for example Page 21 of \url{https://stanford.edu/class/ee363/lectures/inv-sub.pdf}.}, then the controllability of \eqref{sys2} is equivalent to \eqref{sys1}, so the conclusions drawn on the controllability of the new system \eqref{sys2} also hold for the original system \eqref{sys1}. 

To obtain structural controllability results, define $\omega = e^{\frac{2\pi j }{n}}$ characterising the Fourier matrix $F_n$ \cite{marshall, olson}
\begin{align*}
F_n^* = \frac{1}{\sqrt{n}}  \begin{bmatrix} 1 & 1 & 1  & \dots & 1 \\
1 & \omega  & \omega^2 & \dots & \omega^{n-1} \\
1 & \omega^2 & \omega^4 & \dots & \omega^{2(n-1)} \\
\vdots & \vdots & \vdots & & \vdots \\
1 & \omega^{n-1} & \omega^{2(n-1)} & \dots & \omega^{(n-1)(n-1)} 
\end{bmatrix}.
\end{align*}
Then, because $\hat{A}$ is block circulant \cite{marshall, olson}, it can be block-diagonalised by the transformation $\tilde{x} =  (F_n^*\otimes I_n) \hat{x}$ into
\begin{align}\label{dyns_tilde}
\dot{\tilde{x}}(t) &  = \begin{bmatrix}D_1 & & &  \\ & D_2 & & \\ & & \ddots & \\ & & & D_n\end{bmatrix} \tilde{x}(t) 
\\
 & + \frac{1}{\sqrt{n}}\begin{bmatrix}B_1 \\ B_1 \\ \vdots \\ B_1\end{bmatrix}\hat{u}(t) + (F_n^* \otimes I_2)^{-1}d(t), \nonumber
\end{align}
where  $B_1 = [0,1]^T$ and 
\begin{align}
D_i = A_1 + A_2 \omega^{(n-1)(i-1)}, i = 1, 2, \dots, n.
\end{align}
 
Focusing on the dynamics of the first block of this system, these collapse down to
 \begin{align}\label{mode}
 \frac{d}{dt} \begin{bmatrix} \tilde{x}_{11} \\ \tilde{x}_{12}\end{bmatrix}
   & = \begin{bmatrix} 0 & 0 \\ \alpha_1 & -\alpha_2 +\alpha_3\end{bmatrix} \begin{bmatrix} \tilde{x}_{11} \\ \tilde{x}_{12}\end{bmatrix} 
 + \begin{bmatrix}0 \\ \frac{1}{\sqrt{n}}\end{bmatrix}\hat{u}(t)
 \\ &  + \frac{1}{\sqrt{n}} \begin{bmatrix}1 \\ 1 \end{bmatrix} \sum^{n}_{i = 1} d^v_{i}(t)\nonumber
 \end{align}
 where $(F_n^* \otimes I_2)^{-1} = (F_n \otimes I_2)$.
The first equation of \eqref{mode} (being the sum of the vehicle spacings around the ring) is an uncontrollable mode of this system, since it is unaffected by the control action $u(t)$. Crucially, this mode is also only marginally stable, and so it acts as in integrator on the the sum of the velocity disturbances
\begin{align}\label{weiner}
\dot{\tilde{x}}_{11}(t) = \frac{1}{\sqrt{n}}\sum_{i = 1}^{n} d^v_{i} (t)
\end{align}
or
\begin{align}
\tilde{x}_{11}(T)- \tilde{x}_{11}(0)= {\frac{1}{\sqrt{n}}}\int^T_{ 0}\sum_{i = 1}^{n} d^v_{i} (t)~dt.
\end{align}
This shows that the uncontrollable mode is simply the integral of the velocity disturbances.

\end{proof}

We remark that~\eqref{mode} corresponds to the first mode of the transformed system (which is different from vehicle 1). The consequences of this result are best illustrated by assuming that the velocity disturbances are zero mean and Gaussian. In this case, the uncontrollable mode simply integrates the noise, with it being impossible to unwind this effect through feedback control since the mode is also uncontrollable. So, unless the net disturbance condition \eqref{noise_cond} of the theorem is satisfied, then the ring length will undergo a random walk with its variance growing linearly in $t$ since \eqref{weiner} is a Weiner process. In the limit $t \to \infty$, this mode will diverge almost surely under any control action.


\begin{remark}
The equation \eqref{mode} shows that disturbances in the acceleration do not enter into the uncontrollable, marginally stable mode of concern, and hence these disturbances have little impact upon the stability of the system (as illustrated in Figure \ref{fig:fig_pert_v}). In practise, traffic flow rings will be subject to acceleration disturbances (e.g. from wind gusts or from random perturbations in the driver's accelerator etc.) which need to be regulated. The above theorem shows that there always exists some control law that can achieve this. In this way, this work can partly explain the success of the experiment in \cite{exp} since it was in an open environment, as well as the numerical comparison in \cite{yang}.
\end{remark}



\begin{remark}
From the duality of linear systems, the uncontrollability result of Theorem \ref{prop1} can be translated into an unobservability result when measurements of the preceding cars relative velocity are taken, since, in this case, the output matrix satisfies $C = B^T$.  As such, any observer will also suffer issues associated with an unobservable mode integrating up disturbances. Similar issues caused by an integrator were also encountered in  the observer design for super-capacitor energy storage devices \cite{drummond}.
\end{remark}

\begin{remark}
Theorem \ref{prop1}'s issue could be resolved if the autonomous vehicle could have control over its velocity directly. Currently, the control is through the acceleration of the autonomous vehicle.  
\end{remark}

\begin{remark}
It is noted that there is nothing special about the actuator being an autonomous vehicle, as the control policy could equally be applied by a human driver. However, autonomous vehicle control action could be applied more accurately and at a higher frequency than human drivers in theory. That being said, the question of how a human driver would react to being regulated by autonomous systems is not well-understood. 
\end{remark}



Alternatively, the uncontrollable mode could simply be identified from the derivative of a vehicle's spacing with the assumed disturbances
\begin{align}\label{disturbed_s_dot}
\dot{\tilde{s}}_i(t) = \tilde{v}_{i-1}(t)-\tilde{v}_i(t) + d^v_i.
\end{align}
Summing \eqref{disturbed_s_dot} over all the vehicles gives
\begin{align}
\sum^n_{i = 1}\dot{\tilde{s}}_i(t) = \sum^n_{i = 1}\left(\tilde{v}_{i-1}(t)-\tilde{v}_{i}(t) + d^v_i\right)= \sum^n_{i = 1} d^v_i,
\end{align}
generating the uncontrollable, marginally stable mode of \eqref{weiner}. The strength of considering the decomposition exploited in Theorem \ref{prop1} is that it shows that this uncontrollable mode is the only one in the system and also that acceleration disturbances can be easily rejected by feedback control. In this way, Theorem \ref{prop1} completely captures the influence of disturbances on the system.

\begin{figure*}[ht]
\begin{subfigure}{.5\textwidth}
  \centering
  \includegraphics[width=.65\linewidth]{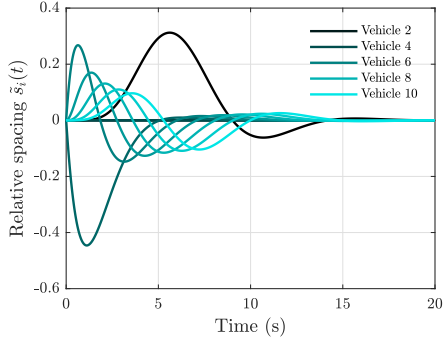}  
  \caption{Relative spacings of each car in the ring.}
  \label{fig:free_s}
\end{subfigure}
\begin{subfigure}{.5\textwidth}
  \centering
  \includegraphics[width=.65\linewidth]{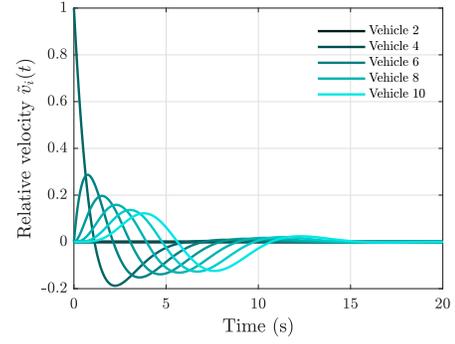}  
  \caption{Relative velocities of each car in the ring.}
  \label{fig:free_v}
\end{subfigure}
\caption{Response of the 10 vehicles in the traffic flow ring controlled by an autonomous vehicle via the policy \eqref{u} from an initial condition of of $\tilde{v}_5(0) = 1$ and no disturbances acting $d^a = w_v = 0$.}
\label{fig:fig_free}
\end{figure*}

\begin{figure*}[ht]
\begin{subfigure}{.5\textwidth}
  \centering
  \includegraphics[width=.65\linewidth]{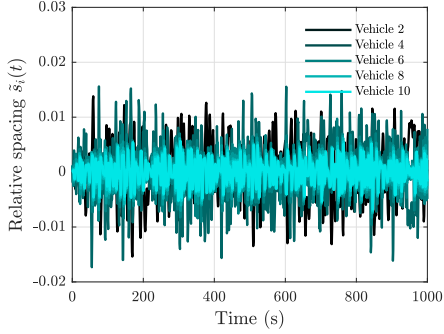}  
  \caption{Relative spacings of each car in the ring.}
  \label{fig:pert_v_s}
\end{subfigure}
\begin{subfigure}{.5\textwidth}
  \centering
  \includegraphics[width=.65\linewidth]{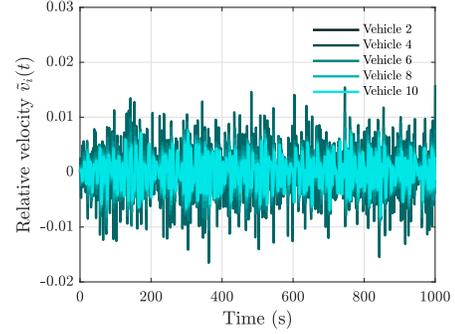}  
  \caption{Relative velocities of each car in the ring.}
  \label{fig:pert_v_v}
\end{subfigure}
\caption{Response of the 10 vehicles in the traffic flow ring controlled by an autonomous vehicle via the policy \eqref{u} from $x(0) = 0$ when under Gaussian, zero mean acceleration  disturbances of variance 1 only applied to $w_5^a(t)$. The response is effectively controlled.}
\label{fig:fig_pert_v}
\end{figure*}

 \begin{table}
 \caption{Parameters for the simulations of Figures \ref{fig:fig_free}-\ref{fig:fig_pert_s}.}
\begin{tabular}{c|l | l} 
\toprule
 & \multicolumn{1}{c}{Model parameters.} & \multicolumn{1}{c}{Value.} \\
\hline
$n$ & Number of cars in the ring. & $10$\\
$N_{\text{dist}}$ & Car with disturbance acting upon it. & $5$\\
$\alpha$ & Factor on the spacing. & $0.6$ \\
$\beta$ & Factor on the velocity.  &$0.9$ \\
$s^*$ & Equilibrium spacing. & 20 m \\
$v^*$ & Equilibrium velocity. &  m s$^{-1}$ \\
$s_{\text{st}}$ & Min. spacing of the vehicles. & 5 m \\
$s_{\text{go}}$ & Max. spacing of the vehicles. & 35 m \\
$v_{\text{max}}$ & Max. velocity of the vehicle. & 30 m s$^{-1}$ \\
$F(\cdot,\cdot,\cdot)$ &Forcing function & N. \\
$\alpha_1$ & $\partial F/\partial s$ & 1.5708 N m$^{-1}$ \\
$\alpha_2$ & $\partial F/\partial v-\partial F/\partial \dot{s}$ & 1.5 N s m$^{-1}$  \\
$\alpha_3$ & $ \partial F/\partial \dot{s}$ & 0.9 N s m$^{-1}$ \\
$d^v$ & Velocity disturbances. &   \\
$d^a$ & Acceleration disturbances. &  \\
\bottomrule 
\end{tabular}\label{tab:params2}
\end{table}

\section{Numerical experiments}\label{sec:sim}
To illustrate the implications of Theorem~\ref{prop1}, a traffic flow ring modelled by \eqref{sys} and \eqref{sys1} was simulated with disturbances in both  velocity and acceleration. For these simulations, the parameters of Table 1 were used (being similar to those used in \cite{jin}), corresponding to 10 vehicles in the ring with an average spacing of 20 m. The following simple control law was used for the autonomous vehicles' acceleration
\begin{align}\label{u}
u(t) = \sum^5_{i = 1}\gamma_i\tilde{s}_i(t) + \lambda_i \tilde{v}_i,
\end{align}
with the control weightings $\gamma \in \mathbb{R}^5$, $\lambda \in \mathbb{R}^5$ set to unity $\gamma_i =\lambda_i = 1$. In this way, the autonomous vehicle's control was determined from the sum of the relative spacing and velocity of the five preceding vehicles. Whilst being a simple enough feedback, it is emphasized that the statement of Theorem \ref{prop1} being tested by the numerical simulations is independent of the applied control. It is also noted that, due to the uncontrollable mode, there does not exist a positive definite solution to the Riccatti equation for obtaining an LQR gain for this system, and so special care has  be taken in choosing an optimal controller.

\begin{figure*}[ht]
\begin{subfigure}{.5\textwidth}
  \centering
  \includegraphics[width=.65\linewidth]{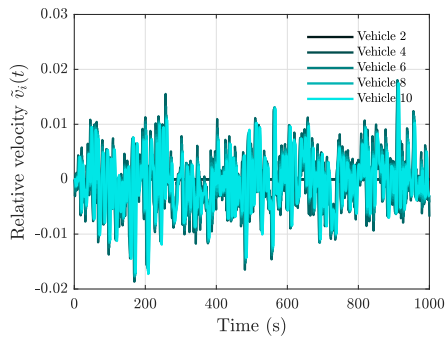}  
  \caption{Relative velocities of each car in the ring.}
  \label{fig:pert_s_s}
\end{subfigure}
\begin{subfigure}{.5\textwidth}
  \centering
  \includegraphics[width=.65\linewidth]{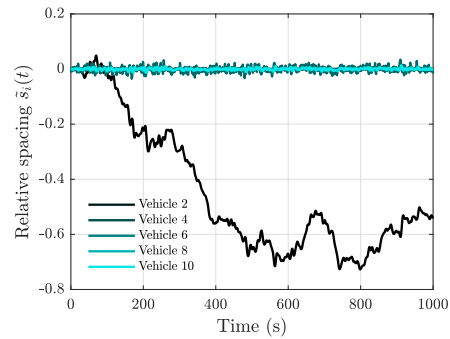}  
  \caption{Relative displacements of each car in the ring.}
  \label{fig:pert_s_v}
\end{subfigure}
\caption{Response of the 10 vehicles in the traffic flow ring controlled by an autonomous vehicle via the policy \eqref{u} from $x(0) = 0$ when under Gaussian, zero mean velocity disturbances of variance 1 only applied to vehicle 5 $w_5^v(t)$. The disturbances are integrated inside an uncontrollable mode of the system leading to a random walk.}
\label{fig:fig_pert_s}
\end{figure*}
The results of the simulations are shown in Figures \ref{fig:fig_free}-\ref{fig:fig_pert_s}, with Figure \ref{fig:fig_free} showing the response from an initial perturbation in the velocity of car 5 $\tilde{v}_5(0) = 1$ and no disturbances applied, Figure \ref{fig:fig_pert_v} showing the response of the ring when subject to disturbances in the acceleration only and Figure \ref{fig:fig_pert_s} the response subject to disturbances in the velocity only. For both cases, the noise was taken to be zero-mean Gaussian with variance 1.

As predicted by Theorem \ref{prop1}, when the velocity perturbations $d^v(t)$ are non-zero, then the disturbance enters an uncontrollable mode (corresponding to the total relative spacings of the vehicles) whereupon it is integrated, resulting in the vehicles undergoing a random walk. In contrast, both the free response (Figure \ref{fig:fig_free}) and the impact of acceleration disturbances $d^a(t)$ (Figure \ref{fig:fig_pert_v}) were rather benign. Velocity disturbance were the key to destabilising the system, agreeing with the conclusion of Theorem \ref{prop1}.

\section{Conclusion}
\label{sec:con}
 It was shown that disturbances acting upon the velocity of any vehicle in a traffic flow ring with autonomous vehicle to control will be integrated up inside an uncontrollable mode. This mode is defined by the sum of the vehicle's relative spacings and cannot be unwound using feedback. To illustrate how this effect may become an issue, it implies that if the velocity disturbance was assumed to be zero mean and Gaussian, then the vehicles will undergo the undesirable response of a random walk which can not be stabilised. In contrast, acceleration disturbances do not enter this uncontrollable mode and hence are not an issue for stabilization. In this way, the non-impact of disturbances in recent autonomous vehicle control experiments could be explained. However, this work also implies that the application of velocity disturbances can destabilise any traffic flow ring controlled by an autonomous vehicle.

\bibliographystyle{IEEEtran}
\bibliography{bibliog}

\end{document}